\documentclass[envcountsame]{fsttcs-ps}

\usepackage{gastex}
\usepackage[utf8x]{inputenc}
\usepackage[T1]{fontenc}
\usepackage{amsmath,amssymb}
\usepackage{xspace}
\usepackage{subfigure}
\usepackage{wasysym}
\usepackage{array}
\usepackage{multicol}

\newenvironment{example}{\theoremlike{Example}}{\par\medskip}

\newcommand{\until}{{Until}\xspace}
\newcommand{\since}{{Since}\xspace}
\newcommand{\tomorrow}{{Next}\xspace}
\newcommand{\eventually}{{Eventually}\xspace}
\newcommand{\always}{{Always}\xspace}

\newcommand{\A}{\mathcal A}
\newcommand{\B}{\mathcal B}
\newcommand{\firstautomaton}{\A_1}
\newcommand{\secondautomaton}{\A_2}
\newcommand{\AP}{\mathrm{AP}}
\newcommand{\J}{J}
\newcommand{\K}{K}
\newcommand{\Until}{{\mathrel{\mathcal U}}}
\newcommand{\Since}{{\mathrel{\mathcal S}}}
\newcommand{\ie}{\textit{i.e.}\xspace}

\newcommand{\truth}{\ensuremath{v}}
\renewcommand{\phi}{\varphi}

\newcommand{\neXt}{\mathop{\mathcal X}}
\newcommand{\Always}{\mathop{\mathcal G}}
\newcommand{\Eventually}{\mathop{\mathcal F}}
\newcommand{\FO}{\mathrm{FO}\xspace}
\newcommand{\pspace}{\textsc{PSpace}\xspace}
\newcommand{\expspace}{\textsc{ExpSpace}\xspace}

\newcommand{\qzero}{q_0}
\newcommand{\qone}{q_1}
\newcommand{\qtwo}{q_2}
\newcommand{\qhole}{q_3}
\newcommand{\qfour}{q_4}
\newcommand{\qfive}{q_5}
\newcommand{\qsix}{q_6}
\newcommand{\qseven}{q_7}
\newcommand{\qeight}{q_8}
\newcommand{\qnine}{q_9}

\newcommand{\shuffle}{\mathop{sh}}



\begin{document}

\title{Automata and temporal logic over arbitrary linear time}
\author{Julien Cristau}
\affiliation{
LIAFA --- CNRS \& Université Paris 7
}

\runningtitle{Automata and temporal logic over arbitrary linear time}
\runningauthors{Julien Cristau}

\begin{abstract}
Linear temporal logic was introduced in order to reason about reactive
systems.  It is often considered with respect to infinite words, to specify
the behaviour of long-running systems.  One can consider more general models
for linear time, using words indexed by arbitrary linear orderings.  We
investigate the connections between temporal logic and automata on linear
orderings, as introduced by Bruyère and Carton.  We provide a doubly
exponential procedure to compute from any LTL formula with \until, \since, and
the Stavi connectives an automaton that decides whether that formula holds on
the input word.  In particular, since the emptiness problem for these automata
is decidable, this transformation gives a decision procedure for the
satisfiability of the logic.
\end{abstract}

\section{Introduction}

Temporal logic, in particular LTL, was proposed by Pnueli to specify the
behaviour of reactive systems~\cite{DBLP:conf/focs/Pnueli77}.  The model of
time usually considered is the ordered set of natural numbers, and executions
of the system are seen as infinite words on some set of atomic propositions.
This logic was shown to have the same expressive power as the first order
logic of order~\cite{Kamp}, but it provides a more convenient formalism to express
verification properties.  It is also more tractable: while the satisfiability
problem of $\FO$ is non-elementary~\cite{stockmeyer74}, it was shown
in~\cite{DBLP:journals/jacm/SistlaC85} that the decision problem of LTL with
\until and \since on $\omega$-words is \pspace-complete.  This logic has also
strong ties with automata, with important work to provide efficient
translations to Büchi automata, e.g.~\cite{DBLP:conf/cav/GastinO01}.

Within this time model, a number of extensions of the logic and the automata
model have been studied.  But one can also consider more general models of
time: general linear time could be useful in different settings, including
concurrency, asynchronous communication, and others, where using the set of
integers can be too simplistic.  Possible choices include ordinals, the reals,
or even arbitrary linear orderings.  In terms of expressivity, while LTL with
\until and \since is expressively complete (\ie equivalent to $\FO$) on
Dedekind-complete orderings (which includes the ordering of the reals as
well as all ordinals), this does not hold in the general case.  Two more
connectives, the future and past Stavi operators, are necessary to handle
gaps~\cite{DBLP:conf/popl/GabbayPSS80} when considering arbitrary linear
orderings.

Over ordinals, LTL with \until and \since has been shown to have a
\pspace-complete satisfiability problem~\cite{DBLP:conf/lpar/DemriR07}.  Over
the ordering of the real numbers, satisfiability of LTL with until and since
is \pspace-complete, but satisfiability of MSO is undecidable.  Over
general linear time, first order logic has been shown to be decidable, as well
as universal monadic second order logic.  Reynolds shows
in~\cite{DBLP:journals/jcss/Reynolds03} that the satisfiability problem of
temporal logic with only the \until connective is also \pspace-complete,
and conjectures that this might stay true when adding the \since
connective.  The upper bound in~\cite{DBLP:conf/lpar/DemriR07} is obtained by
reducing the satisfiability of LTL formulae to the accessibility problem in an
appropriate automata model, accepting words indexed by ordinals.
In this paper, we focus on the general case of arbitrary linear orderings,
using the full logic with \until, \since and both Stavi connectives.
Our aim is to investigate the connections between LTL and automata in this
setting.

Automata on linear orderings were introduced by Bruyère and
Carton~\cite{DBLP:conf/mfcs/BruyereC01}.  This model extends traditional finite
automata using ``limit'' transitions to handle positions with no successor
or predecessor, furthering Büchi's model of automata on words of ordinal
length~\cite{buchiordinals}.  Carton showed in~\cite{DBLP:conf/mfcs/Carton02}
that accessibility over scattered ordering is decidable in polynomial time,
and in~\cite{DBLP:journals/ijfcs/RispalC05} it was shown that these automata
can be complemented over countable scattered linear orderings.  The
accessibility result can be extended to arbitrary
orderings~\cite{cartonprcomm}.

From any formula in this logic, we define an automaton which determines
whether the formula holds on its input word.  Satisfiability of the formula is
reduced to accessibility in this automaton, and that way we get decidability
of the satisfiability problem of LTL with \until, \since and the Stavi
operators for any rational subclass.

Section~\ref{s:defs} presents some definitions about linear orderings, linear
temporal logic, and the model of automata used.
Section~\ref{s:trans} introduces our main result, an algorithm to translate
any LTL formula into a corresponding automaton.
Section~\ref{s:discuss} discusses the expressivity of the logic and automata
considered, and looks at some natural fragments.

\section{Definitions}
\label{s:defs}
\subsection{Linear orderings}

We first recall some basic definitions about orderings, and introduce some
notations.  For a complete introduction to linear orderings, the reader is
referred to~\cite{Rosenstein82}.
A \emph{linear ordering} $\J$ is a totally ordered set $(J,<)$ (considered
modulo isomorphism).
The sets of integers ($\omega$), of rational numbers ($\eta$), and of real
numbers with the usual orderings are all linear orderings.

Let $\J$ and $\K$ be two linear orderings.  One defines the reversed ordering
$-\J$ as the ordering obtained by reversing the relation $<$ in $\J$, and
the ordering $\J+\K$ as the disjoint union $J\sqcup K$ extended with $j<k$ for
any $j\in J$ and $k\in K$.
For example, $-\omega$ is the ordering of negative integers.  $-\omega+\omega$
is the usual ordering of $\mathbb{Z}$, also denoted by $\zeta$.

A non-empty subset $\K$ of an ordering $\J$ is an \emph{interval} if for any
$i<j<k$ in $J$, if $i\in \K$ and $k\in \K$ then $j\in \K$.
In order to define the runs of an automaton, we use the notion of cut.
A \emph{cut} of an ordering $\J$ is a partition $(K,L)$ of $J$ such that for
any $k\in K$ and $l\in L$, $k<l$.
We denote by $\hat\J$ the set of cuts of $\J$.  This set is equipped with the
order defined by $(K_1,L_1)<(K_2,L_2)$ if $K_1\subsetneq K_2$.  This ordering
can be extended to $\J\cup\hat\J$ in a natural way ($(K,L)<j$ iff $j\in L$).
Notice that $\hat\J$ always has a smallest and a biggest element, respectively
$c_{\min} = (\emptyset,J)$ and $c_{\max} = (J,\emptyset)$.
For example, the set of cuts of the finite ordering $\{0,1,\dots,n-1\}$ is the
ordering $\{0,1,\dots,n\}$, and the set of cuts of $\omega$ is $\omega+1$.

For any element $j$ of $\J$, there are two successive cuts $c_j^-$ and
$c_j^+$, respectively $(\{i\in J\mid i<j\},\{i\in J\mid j\le i\})$ and
$(\{i\in J\mid i\le j\},\{i\in J\mid j<i\})$.    A \emph{gap} in an ordering
$J$ is a cut $c$ which is not an extremity ($c_{\max}$ or $c_{\min}$), and has
neither a successor nor a predecessor.

Given an alphabet $\Sigma$, a \emph{word} of length $\J$ is a sequence
$(a_j)_{j\in \J}$ of elements of $\Sigma$ indexed by $\J$.
For example, 
$(ab)^\omega$ is a word of length $\omega$;
the sequence $ab^\omega ab^\omega a$ is a word of length $\omega+\omega+1$,
and
$(ab^\omega)^\omega$ is a word of length $\omega^2$.

\subsection{Temporal logic}

We use words over linear orderings to model the behaviour of systems
over linear time.  To express properties of these systems, we consider
linear temporal logic.
The set of LTL formulae is defined by the following grammar, where $p$ ranges
over a set $\AP$ of atomic propositions:
$\quad\varphi ::= p \ |\  \neg \varphi \ |\  \varphi \vee \varphi \ |\  \varphi
\Until \varphi \ |\  \varphi \Since \varphi \ |\  \varphi \Until' \varphi \ |\
\varphi \Since' \varphi
$

Besides the usual boolean operators, we have four temporal connectives.  The
$\Until$ connective is called ``\until'', and $\Since$ is called ``\since''.
$\Until'$ and $\Since'$ are respectively the future and past Stavi
connectives.  Other usual connectives such as ``\tomorrow'' ($\neXt$),
``\eventually'' ($\Eventually$), ``\always'' ($\Always$) can be defined using
these, as we see below.



These formulae are interpreted on words over the alphabet $2^{\AP}$.
A letter in those words is the set of atomic propositions that hold at the
corresponding position.
Let $x = (x_j)_{j\in \J}$ a word of length $\J$.  A formula $\varphi$ is
evaluated at a particular position $i$ in $x$; we say that $\varphi$ holds at
position $i$ in $x$, and we write $x,i\models\varphi$, using the following
semantics:
\begin{eqnarray*}
x,i\models p & \text{ if } & p\in x_i \\
x,i\models\neg\psi & \text{ if } & x,i\not\models\psi \\
x,i\models\psi_1\vee\psi_2 &\text{ if }& x,i\models\psi_1 \text{ or }
x,i\models\psi_2 \\
x,i\models\psi_1\Until\psi_2 & \text{ if } & \text{there exists $j>i$ such that
$x,j\models\psi_2$,} \\
& & \text{and for any $k$ such that $i<k<j$, we have $x,k\models\psi_1$} \\
x,i\models\psi_1\Since\psi_2 & \text{ if } & -x,i\models\psi_1\Until\psi_2
\text{ where $-x$ is the reversed word $(a_j)_{j\in -\J}$} \\
x,i\models\psi_1\Until'\psi_2 & \text{ if } & \text{there exists a gap $c\in\hat\J$
verifying three properties:} \\
& (1) & x,j\models \psi_1\text{ for any position $j$ such that $i<j<c$}\\
& (2) & \text{there is no interval starting at $c$ where $\psi_1$
is always true} \\
& & \text{(\ie $\forall c<k \ \exists c<j<k\ x,j\models\neg\psi_1$), and } \\
& (3) & \text{$\psi_2$ is always true in some interval starting at $c$} \\
x,i\models\psi_1\Since'\psi_2 & \text{ if } & -x,i\models\psi_1\Until'\psi_2
\text{ (it is the corresponding past connective)}
\end{eqnarray*}

Note that we use a ``strict'' semantic for the \until operator, contrary to a
common definition, which would be:
$$
x,i\models \psi_1\Until^{ns}\psi_2 \ \ \ \text{ if }\ \ \ \text{there exists } j\ge i \text{
such that }
x,j\models\psi_2 \text{ and } x,k\models\psi_1 \text{ for any } i\le
k<j.$$

\noindent In the strict version, the current position $i$ is not considered
for either the $\psi_1$ or the $\psi_2$ part of the definition.  Using the
strict or non-strict version makes no difference when considering
$\omega$-words, but in the case of arbitrary orderings, the strict \until is
more powerful, as noted by Reynolds in~\cite{DBLP:journals/jcss/Reynolds03}.

The formula ``\tomorrow $\varphi$'', or $\neXt\varphi$, is equivalent to
$\bot\Until\varphi$.  ``Eventually $\varphi$'', noted $\Eventually\varphi$, is
$\varphi\vee(\top\Until\varphi)$, and ``always $\varphi$'', noted
$\Always\phi$, can be expressed as $\neg(\Eventually(\neg\varphi))$.


Given a word $x$ of length $\J$, the \emph{truth word} of $\varphi$ on $x$ is
the word $\truth_{\varphi}(x)$ of length $\J$ over the alphabet $\{0,1\}$
where the position $j$ is labelled by $1$ iff $x,j\models\varphi$.  A formula
is \emph{valid} if its truth word on any input only has ones.  A formula is
\emph{satisfiable} if there exists an input word such that the truth word
contains a one.

Consider the formula $\varphi = \neg a\wedge(\Always\neg\neXt a)$, with
$\AP = \{a\}$.  If $x = (a\emptyset)^\omega$ (where $a$ stands for $\{a\}$),
then $\truth_\varphi(x) = 0^\omega$ (at every position, either $a$ is true or
$a$ is true in the successor).  On the other hand, if $x =
a\emptyset^{\omega}a\emptyset^\omega a$, then $\truth_\varphi(x) = 01^\omega
01^\omega 0$: at positions 0, $\omega$ and at the last position, $a$ is true
so the formula doesn't hold; at all other positions, $a$ is false, and there
is no position in the input word where $\neXt a$ holds.

The \emph{satisfiability problem} for a formula $\phi$ consists in deciding
whether there exists a word $w$ and a position $i$ in $w$ such that
$w,i\models \phi$.
As $\FO$ is decidable, and every LTL formula can be expressed using first
order, satisfiability of LTL is decidable.  Note however that in terms of
complexity $\FO$ is already non-elementary on finite
words~\cite{stockmeyer74}, which is not true of LTL.

\subsection{Automata}

On infinite words, Büchi automata can be used to decide satisfiability of LTL
formulae.  In the case of words over linear orderings, a model of automata has
been introduced in~\cite{DBLP:conf/mfcs/BruyereC01}.  Instead of accepting or
rejecting each input word, as in the case of $\omega$-words, we use these
automata to compute the truth words corresponding to an LTL formula.  Our
model of automata thus has an output letter on each transition, so they are
actually letter-to-letter transducers, which make composition easier
(see Section~\ref{s:composition}).

An \emph{automaton} is a tuple $\A = (Q,\Sigma,\Gamma,\delta,I,F)$ where
$Q$ is a finite set of states,
$\Sigma$ is a finite input alphabet,
$\Gamma$ is a finite output alphabet,
$I$ and $F$ are subsets of $Q$, respectively the set of initial and
final states,
and $\delta\subseteq (Q\times\Sigma\times\Gamma\times Q) \cup (2^{Q}\times Q)
\cup (Q\times 2^{Q})$ is the set of transitions.  We note:
\begin{itemize}
\item $p\xrightarrow{a|b}q$ if $(p,a,b,q)\in\delta$ (\emph{successor}
transition)
\item $P\rightarrow q$ if $(P,q)\in\delta$ (\emph{left limit} transition)
\item $q\rightarrow P$ if $(q,P)\in\delta$ (\emph{right limit} transition).
\end{itemize}

Consider a word $x = (q_j)_{j\in\J}$ over $Q$.
We define the left and right limit sets of $x$ at position $j\in J$ as the
sets of labels that appear arbitrarily close to $j$ (respectively to its left
and to its right).  Formally:
\begin{center}
$\lim_{j^-}x = \{q\in Q\mid\forall k<j\ \exists i\ k<i<j\wedge q_i =q\}$\\
$\lim_{j^+}x = \{q\in Q\mid\forall k>j\ \exists i\ j<i<k\wedge q_i =q\}$
\end{center}
Note that $\lim_{j^-}x$ is non-empty if and only if the transition to $j$ is a
left limit, and similarly for $\lim_{j^+}x$ if the transition from $j$ is a
right limit.  These sets help define the possible limit transitions in a run.

Given an automaton $\A$, an accepting run of $\A$ on a word $x =
(x_j)_{j\in\J}$ is a word $\rho$ of length $\hat\J$ over $Q$ such that:
\begin{itemize}
\item $\rho_{c_{\min}}\in I$ and $\rho_{c_{\max}}\in F$;
\item for each $i\in J$, there exists $y_i\in \Gamma$ such that
$\rho_{c_i^-}\xrightarrow{x_i|y_i} \rho_{c_i^+}$;
\item if $c\in\hat\J$ has no predecessor,
$\lim_{c^-}\rho\rightarrow\rho_c$, and if $c\in\hat\J$ has no successor, $\rho_c\rightarrow \lim_{c^+}\rho$.
\end{itemize}

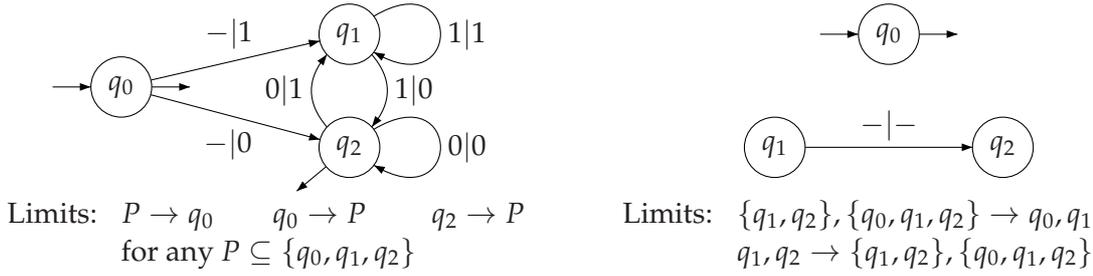
\begin{figure}
\begin{center}
\begin{picture}(80,32)(0,-32)
\node[Nmarks=if](A)(10,-7){$q_0$}
\node(B)(40,0){$q_1$}
\node[Nmarks=f,fangle=220](C)(40,-15){$q_2$}
\drawedge(A,B){$-|1$}
\drawedge[ELside=r](A,C){$-|0$}
\drawloop[loopangle=0](B){$1|1$}
\drawedge[curvedepth=5](B,C){$1|0$}
\drawedge[curvedepth=5](C,B){$0|1$}
\drawloop[loopangle=0](C){$0|0$}

\put(-5,-25){Limits:}
\put(10,-25){$P\rightarrow q_0 \qquad q_0\rightarrow P$ \qquad $q_2\rightarrow
P$}
\put(10,-30){for any $P\subseteq \{q_0,q_1,q_2\}$}
\end{picture}
\begin{picture}(60,32)(0,-32)

\node[Nmarks=if](A)(30,0){$q_0$}
\node(B)(15,-15){$q_1$}
\node(C)(45,-15){$q_2$}
\drawedge(B,C){$-|-$}

\put(-5,-25){Limits:}
\put(10,-25){$\{q_1,q_2\},\{q_0,q_1,q_2\}\rightarrow q_0,q_1$}
\put(10,-30){$q_1,q_2\rightarrow\{q_1,q_2\},\{q_0,q_1,q_2\}$}

\end{picture}
\end{center}
\caption{Example automata}
\label{fig:automata}
\end{figure}

%
%
%
%

\begin{example}
The first automaton in Figure~\ref{fig:automata} outputs $1$ at each position
immediately followed by a $1$ in the input word, and $0$ at other positions.

The second automaton accepts input words whose length is a linear ordering
without first or last element, and without two consecutive elements (\ie
\emph{dense} orderings).  The notation $P\rightarrow q_0,q_1$ means that there
is a transition $P\rightarrow q_0$ and a transition $P\rightarrow q_1$.
\end{example}

In~\cite{DBLP:conf/mfcs/Carton02}, Carton proves that the accessibility
problem on these automata can be solved in polynomial time, when only
considering scattered orderings.  This result can be extended to
arbitrary orderings~\cite{cartonprcomm} as it is done for rational expressions
in~\cite{DBLP:journals/ijfcs/BesC06}.  The idea is to build an automaton
over finite words which simulates the paths in the initial automaton and
remembers their contents.  In order to handle the general case (as opposed to
only scattered orderings), the added operation is called ``shuffle'':
$\shuffle(w_1,\dots,w_n) = \Pi_{j\in J} x_j$ where $J$ is a dense and complete
ordering without a first or last element, partitioned in dense suborderings
$J_1\dots J_n$, such that $x_j = w_i$ if $j\in J_i$.  Looking at automata,
this means that if there are paths from $p_1$ to $q_1$ with content $P_1$,
\dots, from $p_n$ to $q_n$ with content $P_n$, and transitions from
$P_1\cup \dots \cup P_n$ to each $p_i$, transitions from each $q_i$ to
$P_1\cup\dots\cup P_n$, a transition from $p$ to $P_1\cup\dots\cup P_n$ and a
transition from $P_1\cup\dots\cup P_n$ to $q$, then there is a path from $p$
to~$q$.

\section{Translation between formulae and automata}
\label{s:trans}

Over $\omega$-words, problems on temporal logics are commonly solved
using tableau methods~\cite{Wolper85thetableau}, or automata-based
techniques~\cite{DBLP:conf/lics/VardiW86}.  In this work we extend the
correspondence between LTL and automata to words over linear orderings.  Our
main result is Theorem~\ref{th:main}.

\begin{theorem}
\label{th:main}
For every LTL formula $\varphi$, there is an automaton $\A_{\varphi}$ which
given any input word $x$ outputs the truth word $\truth_\varphi(x)$.
\end{theorem}

Moreover, this automaton $\A_{\varphi}$ can be effectively computed, and has a
number of states exponential in the size of $\phi$.  Because we can compute
the product of $\A_{\phi}$ with any given automaton and check for its
emptiness, we get Corollary~\ref{cor:sat}, which states that given a temporal
formula and a rational property (\ie an automaton on words over linear
orderings), we can check whether there exists a model of the formula which is
accepted by the automaton.

\begin{corollary}
\label{cor:sat}
The satisfiability problem for any rational subclass is decidable.
\end{corollary}


The idea is to build $\A_{\varphi}$ by induction on the formula.  We construct
an elementary automaton for each logical connective. We use composition and
product operations to build inductively the automaton of any LTL formula from
elementary automata.  All automata used in this proof have the particular
property that there exists exactly one accepting run for each possible input
word, \ie they are non-deterministic, but also non-ambiguous.  This property
is preserved by composition and product.

The structure of the proof is the following: we define the composition and
product operators on automata, then we present the elementary automata that are
needed to encode logical connectives. Finally, we give the inductive method to
build the automaton corresponding to a formula from elementary ones.


\subsection{Product, composition and elementary automata}
\label{s:composition}

Let $\firstautomaton = (Q_1,\Sigma,\Gamma,\delta_1,I_1,F_1)$ and
$\secondautomaton = (Q_2,\Sigma',\Delta,\delta_2,I_2,F_2)$ be two automata.
The product consists in running both automata with the same input alphabet in
parallel, and outputting the combination of their outputs.  If
$\firstautomaton$'s output alphabet and $\secondautomaton$'s input alphabet
are the same, the composition consists in running $\secondautomaton$ over
$\firstautomaton$'s output.  We use the notation $\pi_1(a,b) = a$ and
$\pi_2(a,b) = b$ for the first and second projections.

\begin{definition}
Suppose that $\firstautomaton$ and $\secondautomaton$ have the same input
alphabet, \ie $\Sigma = \Sigma'$.
The \emph{product} of $\firstautomaton$ and $\secondautomaton$ is the
automaton $\firstautomaton\times\secondautomaton = (Q_1\times
Q_2,\Sigma,\Gamma\times\Delta,\delta,I_1\times I_2,F_1\times F_2)$, where
$\delta$ contains the following transitions:
\begin{itemize}
\item $(q_1,q_2)\xrightarrow{a|b,c}(q'_1,q'_2)$ if $q_1\xrightarrow{a|b}q'_1$
and $q_2\xrightarrow{a|c}q'_2$,
\item $(q_1,q_2)\rightarrow P$ if $q_1\rightarrow \pi_1(P)$ and $q_2\rightarrow
\pi_2(P)$,
\item $P\rightarrow (q_1,q_2)$ if $\pi_1(P)\rightarrow q_1$ and
$\pi_2(P)\rightarrow q_2$.
\end{itemize}
\end{definition}

\begin{definition}
Suppose now that the output alphabet of $\firstautomaton$ is the input
alphabet of $\secondautomaton$, \ie $\Gamma = \Sigma'$.
The \emph{composition}
of $\firstautomaton$ and $\secondautomaton$ is the automaton
$\secondautomaton\circ\firstautomaton = (Q_1\times Q_2, \Sigma, \Delta,
\delta, I_1\times I_2, F_1\times F_2)$.  The transitions in $\delta$ are:
\begin{itemize}
\item $(q_1,q_2)\xrightarrow{a|c}(q'_1,q'_2)$ if $q_1\xrightarrow{a|b}q'_1$
and $q_2\xrightarrow{b|c}q'_2$,
\item $(q_1,q_2)\rightarrow P$ if $q_1\rightarrow \pi_1(P)$ and
$q_2\rightarrow \pi_2(P)$,
\item $P\rightarrow (q_1,q_2)$ if $\pi_1(P)\rightarrow q_1$ and
$\pi_2(P)\rightarrow q_2$.
\end{itemize}
\end{definition}



Recall that LTL formulae are given by $\phi := p\ |\  \neg\phi\ |\
\phi\vee\phi\ |\  \phi\Until\phi\ |\ \phi\Until'\phi\ |\ \phi\Since\phi\ |\
\phi\Since'\phi$.  For each atomic proposition $p$ we construct an automaton
$\A_p$ which, given a word $x$, outputs $\truth_p(x)$.  For each logical
connective of arity $n$, we construct an automaton with input alphabet
$\{0,1\}^n$, and output alphabet $\{0,1\}$.  The input word is the tuple of
truth words of the connective's variables, the output is the truth word of the
complete formula.  For temporal connectives, we only describe the automata
corresponding to $\Until$ and $\Until'$.  For the ``past'' connectives, the
automata are the same with all transitions (successor and limits) reversed,
and initial and final states swapped.

For any $p\in\AP$, the automaton $\A_p$ is
$(\{q\},2^{\AP},\{0,1\},\delta,\{q\},\{q\})$ where $\delta =
\{(q\xrightarrow{a|0}q\mid p\not\in a\}\cup \{q\xrightarrow{a|1}q\mid p\in
a\}\cup\{q\rightarrow \{q\},\{q\}\rightarrow q\}$.  This automaton simply
outputs $1$ at positions where $p$ is true, and $0$ everywhere else.  Note
that the run is uniquely determined by the input word; such a transducer is
called non-ambiguous.

Figures~\ref{fig:not} and~\ref{fig:or} show the automata corresponding to the
negation ($\neg$) and disjunction ($\vee$) connectives.  Their limit
transitions are $\{q\}\rightarrow q$ and $q\rightarrow\{q\}$.  Again, these
automata admit exactly one run for each input word.

\begin{figure}[h]
\begin{center}
\subfigure[Automaton for negation]{
{
\begin{picture}(50,18)(0,-18)
\node[Nmarks=if,iangle=135,fangle=-45](n1)(25,-9.0){$q$}
\drawloop[loopangle=-180](n1){$0|1$}
\drawloop[loopangle=0](n1){$1|0$}
\end{picture}
\label{fig:not}
}}
\subfigure[Automaton for disjunction]{
\begin{picture}(50,18)(0,-18)\nullfont
\node[Nmarks=if,iangle=135,fangle=-45](n1)(25,-9.0){$q$}
\drawloop[loopangle=180](n1){$0,0|0$}
\drawloop[loopangle=0](n1){$0,1|1$ $1,0|1$ $1,1|1$}
\end{picture}
\label{fig:or}
}
\end{center}
\end{figure}



\subsection{Automaton for $\Until$}

The difficulty starts with the ``\until'' connective ($\Until$).  We recall
that $\phi\Until\psi$ holds at position $i$ in a word $w$ if there exists
$j>i$ such that $\psi$ holds at $j$, and such that $\phi$ holds at every
position $k$ such that $i<k<j$.

We build an automaton $\A_{\Until}$ with input alphabet $\{0,1\}²$
(corresponding to the truth value of $\phi$ and $\psi$ at each position), and
output alphabet $\{0,1\}$.  On an input word of the form
$(\truth_{\varphi}(w),\truth_{\psi}(w))$ for some word $w$, we want the output
to be $\truth_{\varphi\Until\psi}(w)$.
Let $J = |w|$, and $c\in \hat J$.  We have five different situations.  For
each of these cases the figure shows an example, with ``$|$''
representing the cut $c$, and each $\bullet$ representing a position in the input
word.

\begin{enumerate}
\setcounter{enumi}{-1}
\item $c$ is followed by a position where $\varphi$ and $\psi$ are true.%
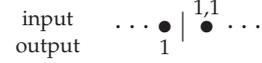
\hfill$
\overset{\text{input}}{\underset{\text{output}}{}}\quad \cdots
\underset{1}{\bullet} 
{|} \overset{1,1}{\bullet} \cdots
$
\item $c = c_j^-$, and $j$ is such that $\varphi$ is false and $\psi$ is
true
.
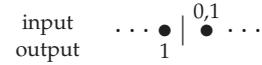
\hfill$
\overset{\text{input}}{\underset{\text{output}}{}}\quad \cdots
\underset{1}{\bullet} 
{|} \overset{0,1}{\bullet} \cdots
$

\item other cases where $\varphi\Until\psi$ is true at $c$.
\hfill$
\overset{\text{input}}{\underset{\text{output}}{}}\quad \cdots
\underset{1}{\bullet} 
{|} \overset{1,-}{\overbrace{\cdots}}
\overset{-,1}{\bullet} \cdots
$

\item $c$ is followed by a position where both $\varphi$ and $\psi$ are false.
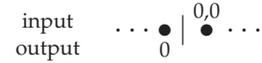
\hfill$
\overset{\text{input}}{\underset{\text{output}}{}}\quad \cdots
\underset{0}{\bullet} 
{|} \overset{0,0}{\bullet} \cdots
$

\item other cases where $\varphi\Until\psi$ is false at $c$.  If $c = c_j^-$
then the input at position $j$ is $(1,0)$.
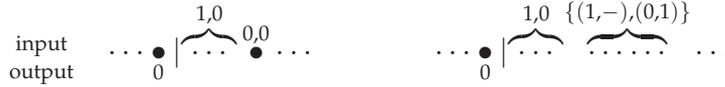
\[
\overset{\text{input}}{\underset{\text{output}}{}}\quad \cdots
\underset{0}{\bullet} 
{|} \overset{1,0}{\overbrace{\cdots}}
\overset{0,0}{\bullet} \cdots
\qquad \qquad
\cdots \underset{0}{\bullet} {|} \overset{1,0}{\overbrace{\cdots}}
\overset{\{(1,-),(0,1)\}}{\overbrace{\cdots\cdots}} \cdots
\]

\end{enumerate}

The structure of the automaton $\A_{\Until}$ and the limit transitions are
given by Figure~\ref{fig:until}.  This automaton has five states $q_0$ to
$q_4$ corresponding to the situations described above.  Given any two states
$q$ and $q'$ there exists a transition $q\rightarrow q'$ except from $q_2$
to $q_3$ or $q_4$ and from $q_4$ to $q_0$, $q_1$ or $q_2$.  The input label of
successor transitions is determined by the origin node: $(1,1)$ for $q_0$,
$(0,1)$ for $q_1$, $(0,0)$ for $q_3$, and $(1,0)$ for $q_2$ and $q_4$.  The
output label is $1$ on transitions leading to $q_0$, $q_1$ or $q_2$, and $0$
on transitions leading to $q_3$ or $q_4$.  All states are initial, while $q_4$ is the
only final state.

%
%

\begin{figure}
\begin{tabular}{m{9cm}m{5cm}}{
\setlength{\unitlength}{.3ex}
\begin{picture}(170,120)(0,-150)\nullfont
\node[NLangle=0.0,Nmarks=i](A)(40.0,-52.0){$q_0$}
\node[NLangle=0.0,Nmarks=i,iangle=0](B)(116.0,-52.0){$q_1$}
\node[NLangle=0.0,Nmarks=i,iangle=0](C)(148.0,-96.0){$q_2$}
\node[NLangle=0.0,Nmarks=i,iangle=200](D)(80.0,-122.0){$q_3$}
\node[NLangle=0.0,Nmarks=if,fangle=180,iangle=135](E)(20.0,-96.0){$q_4$}
\drawloop(A){$1,1/1$}
\drawloop(B){$0,1/1$}
\drawloop[loopangle=270](C){$1,0/1$}
\drawloop[loopangle=270](D){$0,0/0$}
\drawloop[loopangle=270](E){$1,0/0$}
\drawedge[curvedepth=6.0](A,B){$1,1/1$}
\drawedge[curvedepth=6.0](A,C){}
\drawedge[curvedepth=6.0](A,D){}
\drawedge[ELside=r](A,E){$1,1/0$}
\drawedge(B,E){}
\drawedge[ELside=r](D,C){$0,0/1$}
\drawedge[curvedepth=2.0](B,A){}
\drawedge[curvedepth=6.0](B,C){$0,1/1$}
\drawedge[curvedepth=6.0](B,D){}
\drawedge[curvedepth=6.0](C,A){}
\drawedge[curvedepth=4.0](C,B){}
\drawedge[curvedepth=6.0](D,A){}
\drawedge[curvedepth=6.0](D,B){}
\drawedge[curvedepth=6.0](D,E){$0,0/0$}
\drawedge[curvedepth=4.0](E,D){}
\end{picture}
}&
\begin{minipage}{4cm}
{
\input{table_until}
}
\end{minipage}
\end{tabular}
\caption{Automaton for $\Until$}
\label{fig:until}
\end{figure}

\begin{lemma}
Let $\varphi$ and $\psi$ two formulae.  Let $x$ and $y$ be the truth words of
$\varphi$ and $\psi$ on a word $w$ of length $J$.  The output of $\A_{\Until}$
on $(x,y)$ is the truth word of $\varphi\Until\psi$ on $w$.
\end{lemma}

\begin{proof}
Let $\rho$ be the word of length $\hat\J$ on $Q$ defined by
\begin{itemize}
\item if $x_j = y_j = 1$, then $\rho(c_j^-) = q_0$;
\item if $x_j = 0$ and $y_j = 1$ then $\rho(c_j^-) = q_1$;
\item if $x_j = y_j = 0$ then $\rho(c_j^-) = q_3$;
\item otherwise, if there exists $j > c$ such that $y_j = 1$ and for all
$i$ such that $c<i<j$, $x_i = 1$, then $\rho(c) = q_2$;
\item otherwise, $\rho(c) = q_4$.
\end{itemize}
We show that $\rho$ is a run of $\A_{\Until}$, that it is unique, and
that its output is indeed the truth word of $\varphi\Until\psi$ on $w$.

By definition, $\rho$ ends in $q_4$, which is the final state of
$\A_{\Until}$.  Let $c\in \hat J$.  If $\rho(c)$ is $q_0$,  $q_1$
or $q_3$, then $c = c_j^-$ for some $j$ and the successor transition from $c$
to the next cut is allowed by the automaton.  If $\rho(c) = q_2$, and $c =
c_j^-$ for some $j$, then $x_j = 1$ and $y_j = 0$, and $\rho(c_j^+)$ is $q_0$,
$q_1$ or $q_2$.  If $\rho(c_j^-) = q_4$, then similarly $x_j = 1$ and $y_j =
0$, and $\rho(c_j^+)$ can be $q_3$ or $q_4$.  Every successor transition in
$\rho$ is thus allowed by $\A_{\Until}$.

We now need to show the same for limit transitions.
If a left limit transition leads to a cut $c$, then either $\psi$ is true
arbitrarily close to the left of $c$ (in which case the corresponding limit
set contains $q_0$ or $q_1$), or it is always false (and the limit set is
$\{q_2\}$ or a subset of $\{q_3,q_4\}$).  If the limit set contains $q_0$,
$q_1$ or $q_3$, any state for $c$ is allowed.  If it is $\{q_2\}$, the cut $c$
can't be labelled by $q_3$ or $q_4$ without violating the definition of
$\rho$.  Conversely, if the limit set is $\{q_4\}$, $\rho(c)$ is necessarily
$q_3$ or $q_4$.

Let's now consider a right limit transition starting at a cut $c$.  The label
of this cut can only be $q_2$ or $q_4$.  In the first case, $\varphi$ must
be true everywhere in the limit set, which is thus a subset of $\{q_0,q_2\}$.
In the second case, either $\varphi$ is false infinitely often in the limit,
or $\psi$ is always false.  This means that the limit set contains $q_1$ or
$q_3$, or is restricted to $\{q_4\}$.

We now show that a run on $\A_{\Until}$ is uniquely determined by the input
word.
Let $\gamma$ a run of $\A_{\Until}$ on $x,y$.  Because of the constraints on
the successor transitions, a cut $c$ is labelled by $q_0$, $q_1$ or $q_3$ in
$\gamma$ if and only if it is labelled by the same state in $\rho$.

Let's suppose that a cut $c$ is labelled by $q_2$ in $\gamma$.  Since $q_2$ is
not final, there exists $c' > c$ labelled by some other state.  If there is a
first such cut, its label is necessarily $q_0$ or $q_1$ (by a successor
transition from $q_2$ or a limit transition from $\{q_2\}$).  Otherwise, there
is a transition of the form $q_2\rightarrow \{q_0\}$ or $q_2\rightarrow
\{q_0,q_2\}$.  In both cases, $c$ satisfies the condition for cuts labelled by
$q_2$ in the definition of $\rho$.
A similar argument shows that a cut labelled by $q_4$ in $\gamma$ has the same
label in $\rho$.  The run of $\A_{\Until}$ on a given input word is thus
unique.

The last step is to show that the output word is really the truth word of
$\phi\Until\psi$.
Let $j$ an element of $J$.
First, suppose that $w,j\models \phi\Until\psi$.  If $j$ has a successor $k$,
and $\psi$ is true at $k$, then $y_k = 1$, and $\A_{\Until}$ outputs $1$ at
position $j$.  Otherwise, there exists $k>j$ such that $w,k\models\psi$ (i.e.
$y_k = 1$), and $x_\ell = 1$ whenever $j<\ell<k$.  Thus, $c_j^+$ is labelled
with $q_2$, and $\A_{\Until}$ once again outputs $1$ at position $j$.

Similarly, if $w,j\not\models\phi\Until\psi$, there are two cases.
In the first case, $j$ has a successor $k$, and $x_k = y_k = 0$.  This means
that $c_j^+$ is labelled by $q_3$, so the output at position $j$ is $0$.
In the last case, $c_j^+$ is labelled by $q_4$, and once again $\A_\Until$
outputs $0$.
\qed\end{proof}

\subsection{Automaton for the future Stavi connective ($\Until'$)}

Let's recall that $\phi\Until'\psi$ holds at position $i$ if there exists a
gap $c > i$ such that $\phi$ holds at every position $i < j < c$, the property
$\psi$ holds at every position in some interval starting at $x$, and
$\neg\phi$ holds at positions arbitrarily close to $c$ to the right.

The central point in this definition is the gap $c$, which corresponds to
state $\qhole$ in the automaton.  States $\qzero$, $\qone$ and $\qtwo$ follow
the positions, before $\qhole$, where the formula holds.  States $\qfour$,
$\qfive$, $\qsix$, $\qseven$, $\qeight$ follow the positions where the formula
doesn't hold.
If a run reaches $\qzero$, $\qone$ or $\qtwo$, it has to leave this region
through $\qhole$, and all successor transitions until then have input label
$(1,0)$ or $(1,1)$.  The structure of this automaton is depicted in
Figure~\ref{fig:stavi}.
All states except $\qhole$ and $\qnine$ are initial; $\qeight$ and $\qnine$
are final.  Transitions from $\qone$ and $\qseven$ have input label $(1,1)$,
transitions from $\qtwo$ and $\qsix$ have input label $(1,0)$, transitions from
$\qfour$ have input label $(0,0)$, and transitions from $\qfive$ have input
label $(0,1)$.  The output is $1$ for transitions to $\qzero$, $\qone$ and
$\qtwo$, and $0$ for transitions to $\qfour$, $\qfive$, $\qsix$, $\qseven$ and
$\qeight$.

\begin{figure}
\begin{tabular}{ll}{
%
%
%
%

\setlength{\unitlength}{.3ex}

\begin{picture}(150,70)(10,-100)\nullfont
\node[NLangle=0.0](q2)(40.0,-32.0){$\qtwo$}

\node[NLangle=0.0](q1)(88.0,-32.0){$\qone$}

\node[NLangle=0.0](q0)(64.0,-52.0){$\qzero$}

\node[NLangle=0.0](qhole)(144.0,-32.0){$\qhole$}

\node[NLangle=0.0](q7)(44.0,-156.0){$\qseven$}

\node[NLangle=0.0](q6)(108.0,-156.0){$\qsix$}

\node[NLangle=0.0](q5)(44.0,-100.0){$\qfive$}

\node[NLangle=0.0](q4)(108.0,-100.0){$\qfour$}

\node[NLangle=0.0,Nmarks=f](q8)(144.0,-104.0){$\qeight$}

\node[NLangle=0.0,Nmarks=f](q9)(144.0,-64.0){$\qnine$}

\drawedge[curvedepth=6.0](q5,q7){}

\drawedge[curvedepth=6.0](q7,q5){}

\drawedge[curvedepth=4.0](q4,q5){}

\drawedge[curvedepth=4.0](q5,q4){}

\drawedge[curvedepth=4.0](q4,q6){}

\drawedge[curvedepth=4.0](q6,q4){}

\drawedge[curvedepth=4.0](q6,q7){}

\drawedge[curvedepth=4.0](q7,q6){}

\drawedge[curvedepth=4.0](q5,q6){}

\drawedge[curvedepth=4.0](q6,q5){}

\drawedge[curvedepth=8.0](q7,q4){}

\drawedge[curvedepth=8.0](q4,q7){}

\drawedge[curvedepth=8.0](q2,q1){}

\drawedge[curvedepth=8.0](q1,q2){}

\drawedge[curvedepth=-4.0](q2,q0){}

\drawedge[curvedepth=8.0](q1,q0){}

\drawloop[loopangle=-180.0,ELpos=70](q7){1,1/0}

\drawloop[loopangle=180.0,ELpos=30](q5){0,1/0}

\drawloop[loopangle=-0.0](q6){1,0/0}

\drawloop[loopangle=45,ELpos=60](q4){0,0/0}

\drawloop[loopangle=0.0](q1){1,0/1}

\drawloop[loopangle=180.0,ELpos=30](q2){1,1/1}

\drawedge[curvedepth=8.0](q4,q0){}

\drawedge[curvedepth=8.0](q5,q0){}

\drawedge[curvedepth=-4.0](q6,q8){}%

\drawedge[curvedepth=-4.0](q4,q8){}%

\drawedge[curvedepth=-12.0](q5,q8){}%

\drawedge[curvedepth=-8.0](q7,q8){}%

\drawedge[curvedepth=8.0](q5,q2){}

\drawedge[curvedepth=-12.0](q5,q1){}

\drawedge[curvedepth=-8.0](q4,q1){}

\drawedge[curvedepth=15.0](q4,q2){}





\end{picture}

}&
\begin{minipage}{6cm}{
\begin{itemize}
\item $P\rightarrow \qzero,\qone,\qtwo$ if $P\cap \{\qfour,\qfive\} \neq
\emptyset$ or $P\subseteq \{\qzero,\qone,\qtwo\}$
\item $P\rightarrow \qhole$ if $P\subseteq \{\qzero,\qone,\qtwo\}$
\item $P\rightarrow \qfour,\qfive,\qsix,\qseven$ if $P\not\subseteq
\{\qzero,\qone,\qtwo\}$
\item $P\rightarrow \qeight$ if $P\cap \{\qfour,\qfive\} \neq \emptyset$
\item $P\rightarrow \qnine$ if $P\cap \{\qfour,\qfive\} = \emptyset$ and
$P\not\subseteq\{\qzero,\qone,\qtwo\}$
\item $\qzero\rightarrow P$ if $P\subseteq \{\qzero,\qone,\qtwo\}$
\item $\qhole\rightarrow P$ if $P\cap \{\qone,\qfour,\qsix\} = \emptyset$ and $\qfive\in P$
\item $\qeight\rightarrow P$ if $P\cap \{\qfour,\qfive,\qsix,\qseven\} \neq \emptyset$
\item $\qnine\rightarrow P$ if $P\cap \{\qfour,\qfive,\qsix,\qseven\} \neq
\emptyset$ and either $P\cap \{\qfour,\qfive\}=\emptyset$ or $P$ intersects
$\{\qone,\qfour,\qsix\}$
\end{itemize}
}\end{minipage}
\end{tabular}
\caption{Automaton for the future Stavi operator}
\label{fig:stavi}
\end{figure}

We define a labelling $\rho$ of the cuts of a word $w$ on $\{0,1\}²$ using the
states of the automaton in the following way:
\begin{itemize}
\item $\qzero$ has no successor, $\phi\Until'\psi$ is true
\item $\qone$ has an outgoing transition labelled $(1,0)$, $\phi\Until'\psi$ is
true
\item $\qtwo$ has an outgoing transition labelled $(1,1)$, $\phi\Until'\psi$ is
true
\item $\qhole$ is a gap, $\phi\Until'\psi$ is true before it and false
afterwards
\item $\qfour$ has an outgoing transition labelled $(0,0)$, $\phi\Until'\psi$ is
false
\item $\qfive$ has an outgoing transition labelled $(0,1)$, $\phi\Until'\psi$ is
false
\item $\qsix$ has an outgoing transition labelled $(1,0)$, $\phi\Until'\psi$ is
false
\item $\qseven$ has an outgoing transition labelled $(1,1)$, $\phi\Until'\psi$
is false
\item $\qeight$ has no successor, $\phi$ doesn't hold in the left limit if it
has no predecessor, and $\phi\Until'\psi$ is false
\item $\qnine$ is a gap or is the last cut, $\phi\Until'\psi$ is false, and
$\phi$ is true in some interval to the left
\end{itemize}


\begin{lemma}
$\rho$ defines the unique run of the automaton on its input word.  If the
input is $(\truth_{\phi}(w), \truth_{\psi}(w))$ for some word $w$, then the
output of this run is $\truth_{\phi\Until'\psi}(w)$.
\end{lemma}
\begin{proof}
We first show that $\rho$ is a run.  Successor transitions correspond almost
directly to the definitions of the labelling $\rho$, so let's look at limit
transitions.  For left limits, the following cases need to be considered:
\begin{itemize}
\item if a transition $P\rightarrow \qzero$ is taken at a cut $c$, then either
$\phi$ is true in the limit, and so $\phi\Until'\psi$ is too, and $P\subseteq
\{\qzero,\qone,\qtwo\}$, or it's not, and either $\qfour$ or $\qfive$ appear
in the limit
\item the same reasoning applies for $\qone$ and $\qtwo$
\item if $c$ is labelled $\qhole$ then the incoming transition has to come
from a subset of $\{\qzero,\qone,\qtwo\}$ since $\phi\Until'\psi$ is true in
the limit.
\item if a transition $P\rightarrow \qfour$ is used, then $\phi\Until'\psi$ is
not true in the limit (otherwise it would still be true), and so
$P\not\subseteq\{\qzero,\qone,\qtwo\}$; the same applies for $\qfive$,
$\qsix$, $\qseven$, $\qeight$ and $\qnine$
\item if $c$ is a left limit and is labelled $\qeight$ then the incoming
transition comes from a set $P$ intersecting $\{\qfour,\qfive\}$ because
$\neg\phi$ is repeated
\item if $c$ is labelled $\qnine$ then $\qfour$ and $\qfive$ can't appear in
the left limit set ($\phi$ is true)
\end{itemize}
If $c$ is a right limit cut, it can only be labelled 
$\qzero$, $\qhole$, $\qeight$ or $\qnine$.
Here are the possible right-limit transitions:
\begin{itemize}
\item if a right-limit cut $c$ is labelled $\qzero$, the
limit transition has to go to a subset of $\{\qzero,\qone,\qtwo\}$ since
$\phi\Until'\psi$ holds in the limit
\item if $c$ is labelled with $\qhole$, the limit transition to its right
leads necessarily to a set $P$ not including $\qone$, $\qfour$ and $\qsix$
since $\psi$ is always true, and including $\qfive$ because $\neg\phi$ is
repeated
\item if $c$ is labelled $\qeight$ or $\qnine$, the right limit set can't be a
subset of $\{\qzero,\qone,\qtwo\}$ otherwise $c$ would have been labelled
$\qzero$
\item if $c$ is labelled $\qnine$ we have the additional condition that
either $\phi$ holds in the limit (and neither $\qfour$ nor $\qfive$ appears)
or $\psi$ doesn't (and one of $\qone$, $\qfour$ and $\qsix$ is in the limit)
\end{itemize}

The labelling of cuts defined above is thus a path of the automaton, and we
only need to show that it's the only one, using the same method as for the
$\A_{\Until}$.  Moreover, the definition of $\rho$ means that the output
is $1$ whenever $\phi\Until'\psi$ holds, and $0$ at all other positions.%
\qed\end{proof}

\subsection{Construction of $\A_{\varphi}$}

Now that we have the basic blocks for our construction, we can build an
automaton for any formula $\varphi$.  If $\varphi$ is an atomic proposition
$p$, then we have seen how to build $\A_p$ in the previous section.  If
$\varphi = \neg\psi$, then $\A_{\varphi} = \A_{\neg}\circ\A_{\psi}$.
If $\varphi = \psi_1\vee\psi_2$, then $\A_{\varphi} = \A_{\vee}\circ
(\A_{\psi_1}\times \A_{\psi_2})$.  If $\varphi = \psi_1\Until\psi_2$, then
$\A_{\varphi} = \A_{\Until}\circ(\A_{\psi_1}\times\A_{\psi_2})$.
The same can be done for $\Until'$ and for the past connectives.

The number of states of the resulting automaton is the product of the number
of states of all the elementary automata, and is thus exponential in the size
of the formula.  The actual size of the automaton includes limit transitions,
so can be doubly exponential in the size of the formula, if those transitions
are represented explicitly.

To check whether the formula $\phi$ is satisfiable by a model which is
recognized by an automaton $\B$, we can compute the product of the automaton
$\A_{\phi}$ with $\B$, and check whether a transition where $\A_{\phi}$
outputs $1$ is accessible and co-accessible.  This ensures that there exists a
successful run of the product automaton going through that transition, meaning
that the corresponding input word is accepted by $\B$ and there is a position
where $\phi$ holds.  This concludes the proof of Corollary~\ref{cor:sat}.


\section{Discussion}
\label{s:discuss}

\noindent \textbf{Logical characterization of automata.} We have shown that
any LTL, and thus $\FO$, formula can be represented as a non-ambiguous
automaton with output.  But one can also build such an automaton where the
output is the truth word of a property which can't be expressed as a
first-order formula.  The automaton shown on Figure~\ref{fig:gap} outputs 1
whenever ``there is a gap somewhere in the future'' is true; that formula
can't be expressed in $\FO$.  It would be interesting to find a logical
characterisation of the properties that can be expressed using such automata.

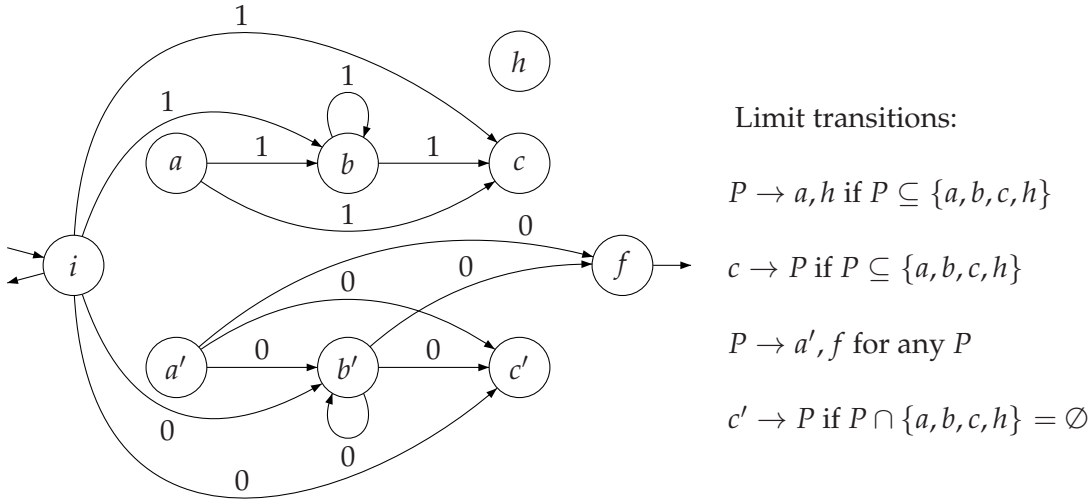
\begin{figure}
\subfigure{
\setlength{\unitlength}{.5ex}
\begin{picture}(80,65)(-10,-65)
\node[Nmarks=if,iangle=165,fangle=195](i)(0,-30){$i$}
\node[Nmarks=f](f)(80,-30){$f$}
\node(a)(15,-15){$a$}
\node(b)(40,-15){$b$}
\node(c)(65,-15){$c$}
\node(h)(65,0){$h$}
\node(a')(15,-45){$a'$}
\node(b')(40,-45){$b'$}
\node(c')(65,-45){$c'$}

\drawloop[loopdiam=6](b){1}
\drawloop[loopdiam=6,loopangle=270](b'){0}
\drawedge[curvedepth=15](i,b){1}
\drawbpedge(i,90,50,c,135,30){1}
\drawedge(a,b){1}
\drawedge[curvedepth=-10](a,c){1}
\drawedge(b,c){1}
\drawedge[curvedepth=-15,ELside=r](i,b'){0}
\drawbpedge(i,270,50,c',225,30){0}
\drawedge(a',b'){0}
\drawedge[curvedepth=10](a',c'){0}
\drawedge(b',c'){0}

\drawedge[curvedepth=10,ELpos=80](a',f){0}
\drawedge[curvedepth=5](b',f){0}
\end{picture}
}
\subfigure{
\begin{picture}(40,50)(0,-50)
\put(20,0){ Limit transitions:}
\put(20,-10){$P \rightarrow a,h$ if $P\subseteq \{a,b,c,h\}$}
\put(20,-20){$c \rightarrow P$ if $P\subseteq \{a,b,c,h\}$}
\put(20,-30){$P \rightarrow a',f$ for any $P$}
\put(20,-40){$c'\rightarrow P$ if $P\cap \{a,b,c,h\} = \emptyset$}

\end{picture}
}
\caption{Automaton checking whether a gap exists in the future}
\label{fig:gap}
\end{figure}

\noindent \textbf{Computational complexity.}  The exact complexity of the
satisfiability problem for LTL on arbitrary orderings remains open.  We give a
2\expspace procedure to compute an automaton from a formula, whose emptiness
can then be checked efficiently.  A classical optimization in similar problems
is to compute the automaton on the fly, which saves a lot of complexity, so an
algorithm using this technique for LTL on arbitrary orderings would be
interesting.

\noindent \textbf{Expressive power.}  On finite and $\omega$-words, LTL
restricted to the unary operators ($\neXt$, $\Eventually$, and their past
counterparts) is equivalent to first-order logic restricted to two variables,
$\FO^2(<,+1)$~\cite{DBLP:journals/iandc/EtessamiVW02}.  Restricting even
further to $\Eventually$ and its reverse, we get a logic expressively
equivalent to $\FO^2(<)$.  In the case of finite words, $\FO^2(<)$ corresponds
to ``partially ordered'' two-way automata~\cite{SchThVo2002}.  The proof of
equivalence between unary temporal logic and $\FO^2$ can be easily extended to
the case of arbitrary linear orderings.  It would be interesting to find such
a correspondence for arbitrary orderings as well, and to see if these
restrictions provide lower complexity results.

\noindent \textbf{Mosaics technique.}  In his work on LTL($\Until$), Reynolds
uses ``mosaics'' to keep track of the subformulas that need to be satisfied in
particular intervals, and to find a decomposition that shows the
satisfiability of the initial formula.  Unfortunately it is not clear if and
how this can be extended to handle a larger fragment of the logic.

\section{Conclusion}

We investigate linear temporal order with \until, \since, and the Stavi
connectives over general linear time, and its relationship with automata over
linear orderings.  We provide a translation from LTL to a class of
non-ambiguous automata with output, giving a 2\expspace procedure to decide
satisfiability of a formula in any rational subclass.

This leaves a number of immediate questions, starting with the actual
complexity for the satisfiability problem for LTL, but also for some of its
fragments, where some operators are excluded.  While the full class of
automata over linear orderings is not closed under
complementation~\cite{DBLP:conf/csr/BedonBCR08}, it might still be possible to
find a logical characterization for some interesting subclasses.

\bibliographystyle{plain}
{\small\bibliography{bib}}
\end{document}